\documentclass[draftcls,onecolumn,12pt]{IEEEtran}
\usepackage{xspace,amsthm,amsmath,amssymb,amsfonts,epsfig,syntonly}
\usepackage{cite,bm,color,url,textcomp,balance}
\usepackage{epstopdf}
\usepackage{empheq}

\usepackage{algorithm}
\usepackage{algorithmicx}
\usepackage{algpseudocode}

\usepackage{graphicx}
\usepackage{subfigure}

\usepackage{amsmath,amssymb,amsthm}
\usepackage{diagbox}
\usepackage{bigints}
\usepackage{stfloats,siunitx}

\newtheorem{lemma}{Lemma}

\newtheorem{proposition}{Proposition}

\newtheorem{remark}{Remark}

\newcounter{TempEqCnt}

\DeclareMathOperator*{\argmin}{arg\,min}

\long\def\symbolfootnote[#1]#2{\begingroup
\def\thefootnote{\fnsymbol{footnote}}
\footnote[#1]{#2}\endgroup}

\psfull

\allowdisplaybreaks[3]

\IEEEoverridecommandlockouts

\begin{document}
\title{Distributed Coordinated Multicell Beamforming for Wireless Cellular Networks Powered by Renewables:
A Stochastic ADMM Approach}

\author{Shuyan Hu, Chongbin Xu,~\IEEEmembership{Member, IEEE}, Xin Wang,~\IEEEmembership{Senior Member, IEEE},\\
Yongwei Huang,~\IEEEmembership{Senior Member, IEEE}, and Shunqing Zhang,~\IEEEmembership{Senior Member, IEEE}
\thanks{This work was supported by the National Natural Science Foundation of China Grants No. 61671154, No. 61501123, and the Innovation Program of Shanghai Municipal Education Commission. Part of the work was presented at ICCS 2016 \cite{conf}.}
\thanks{S. Hu, C. Xu, and X. Wang are with the Shanghai Institute for Advanced Communication and Data Science, Key Laboratory for Information Science of Electromagnetic Waves (MoE), Department of Communication Science and Engineering, Fudan University, Shanghai, China (emails:~\{syhu14,\,chbinxu,\,xwang11\}@fudan.edu.cn).

Y. Huang is with the School of Information Engineering, Guangdong University of Technology, Guangzhou, China
(email: ywhuang@gdut.edu.cn).

S. Zhang is with the Shanghai Institute for Advanced Communication and Data Science, Shanghai University, Shanghai, China
(email: shunqing@shu.edu.cn).}

}

\maketitle

\begin{abstract}
The integration of renewable energy sources (RES) has facilitated efficient and sustainable resource allocation for wireless communication systems.
In this paper, a novel framework is introduced to develop coordinated multicell beamforming (CMBF) design for wireless cellular networks powered by a smart microgrid, where the BSs are equipped with RES harvesting devices and can perform two-way (i.e., buying/selling) energy trading with the main grid.
To this end, new models are put forth to account for the stochastic RES harvesting, two-way energy trading,
and conditional value-at-risk (CVaR) based energy transaction cost.
Capitalizing on these models, we propose a distributed CMBF solution to minimize the grid-wide transaction cost
subject to user quality-of-service (QoS) constraints.
Specifically, relying on state-of-the-art optimization tools, we show that the relevant task can be formulated as a convex problem that is well suited for development of a distributed solver.
To cope with stochastic availability of the RES, the stochastic alternating direction method of multipliers (ADMM) is then leveraged to develop a novel distributed CMBF scheme. It is established that the proposed scheme is guaranteed to yield the optimal CMBF solution, with only local channel state information available at each BS and limited information exchange among the BSs. Numerical results are provided to corroborate the merits of the proposed scheme.
\end{abstract}



\section{Introduction}

With swift developments in computer and communication science, the upcoming fifth-generation (5G) era will witness many uprising mobile services, such as e-banking, e-health, e-learning, and social networking.
Proliferation of smart phones and tablets are driving explosive demands for wireless capacity.
It is predicted that data rate will increase a thousand-fold over the next decade \cite{oss14, ayy16}.
To accommodate such huge data traffic, traditional macro base stations (BSs) have evolved into pico and femto BSs,
where smaller BSs jointly serve overlapping areas to enhance quality-of-service (QoS) for end users in edge areas \cite{hwang13}.
In the resultant heterogeneous networks consisting of overlapping macro/micro/pico cells, coordinated multicell beamforming (CMBF) has emerged as a promising technique, where neighboring BSs jointly design their transmit beamformers to mitigate inter-BS interferences.

CMBF has received growing research interest in the past decade.
The optimal CMBF schemes were proposed to minimize the sum-power of BSs under user signal-to-interference-plus-noise-ratio (SINR) constraints or maximize the minimum user SINR under per BS power constraints in \cite{rash98, hou14, ngu11}.
Based on the uplink-downlink duality, \cite{dah10} developed coordinated beamforming for a multicell multi-antenna wireless system, to
minimize either the total weighted transmitted power or the maximum per-antenna power across the BSs subject to user SINR constraints.
Coordinated multicell beamformers were designed to balance user SINRs to multiple levels in \cite{yuwu15}.
All the existing CMBF schemes in \cite{rash98, hou14, ngu11, dah10, yuwu15} require a central controller
with global channel state information (CSI) and/or global user data sharing.

To avoid the large signalling and backhaul overheads resulting from such centralized schemes,
quite a few distributed CMBF solutions were proposed in \cite{huang11, pen11, tolli11, shen12, ng11},
where each BS devises its own beamformers using its local CSI with the help of limited information exchange among BSs.
Specifically, \cite{huang11} proposed a hierarchical iterative algorithm to jointly optimize downlink beamforming
and power allocation schemes in a distributed manner.  Leveraging the primal and/or dual decomposition techniques, decentralized CMBF schemes were pursued in
\cite{pen11} and \cite{tolli11}. Relying on the alternating direction method of multipliers (ADMM), \cite{shen12} proposed a distributed CMBF scheme which is robust to CSI errors. Game-theory based distributed CMBF scheme was also developed in \cite{ng11}.
All the works in \cite{rash98, hou14, ngu11, yuwu15, dah10, huang11, pen11, tolli11, shen12, ng11}
assumed that the BSs are supplied by persistent energy sources from the conventional power grid.

Insatiable demands for wireless capacity have led to tremendous energy consumption, carbon dioxide ($\text{CO}_2$) emission, and electricity bills for the service providers \cite{ismail11, ismail15}. Driven by the urgent need of energy-efficient and sustainable ``green communications,'' cellular network operators have started developing the ``green'' BSs that can be jointly supplied by the main electric grid as well as the harvested clean and renewable (e.g., solar/wind/thermal) energy \cite{Li15}. It is expected that renewable powered BSs will be widely used in the future cellular systems. In addition, the current power grid infrastructure is also on the verge of migrating from the aging grid to a ``smart'' one. Integration of renewable energy resources (RES) and smart-grid technologies into system designs clearly holds the key to fully exploiting the potential of green communications.
To this end, our recent works \cite{Wan15, Wan16, xiao16, xiao17} leveraged smart-grid capabilities to explore efficient coordinated beamforming designs for coordinated multipoint (CoMP) systems with RES. The centralized CoMP solutions in \cite{Wan15, Wan16, xiao16, xiao17} need to be determined by a central controller that can collect the global CSI, global energy information, as well as all users' data.

In this paper, we address distributed CMBF design for a smart-grid powered coordinated multicell downlink system.
Different from \cite{rash98, hou14, ngu11, yuwu15, dah10, huang11, pen11, tolli11, shen12, ng11}, we assume that the BSs have local RES and can carry out two-way energy trading with the main grid. To account for the unpredictable and nondispatchable nature of RES, a \emph{conditional value-at-risk} (CVaR) cost function is introduced to pursue both efficient and robust resource allocation actions \cite{Yu15}.
Relying on the semidefinite relaxation (SDR) technique, we formulate a convex CMBF problem to minimize the system-wide CVaR-based transaction cost
subject to user quality-of-service (QoS) constraints.
Suppose that the distribution function of the (random) RES amount is unknown; yet, a large set of independent and identically distributed (i.i.d.) observations/realizations are available (e.g., collected from historical measurement and stored in the big database) per BS.
Leveraging the {\em stochastic} ADMM, we develop a systematic approach to obtain the desired CMBF solution in a distributed fashion. It is shown that the proposed algorithm is guaranteed to find the optimal CMBF scheme, with only local CSI at each BS and limited information exchange among BSs.

To the best of our knowledge, our work is the first to leverage the {\em stochastic} ADMM approach in
distributed CMBF design for renewable powered wireless cellular networks.
The main contribution of this paper is three-fold:
i) a CVaR-based transaction cost is introduced in CMBF design to minimize the average energy cost while
effectively avoiding the risk of a very large electricity bill for the RES powered multicell system;
ii) an SDR-based decomposable convex problem reformulation is proposed to facilitate the development of distributed CMBF solution;
and iii) a novel stochastic ADMM approach is developed to find the optimal CMBF scheme in a distributed manner for practical stochastic (uncertain) RES powered cellular network environment.

The rest of the paper is organized as follows. Section \ref{sec:comp} describes the system models, while Section \ref{sec:cvxform} formulates our CVaR-based CMBF problem. The proposed distributed CMBF solution is developed in Section \ref{sec:distcmbf}.
Numerical results are provided in Section \ref{sec:sim}. The paper is concluded in Section \ref{sec:con}.

\section{System Models}\label{sec:comp}

Consider a coordinated multicell downlink system consisting a set of ${\cal I}:= \{1, \ldots, I\}$ different (e.g., macro/pico/femto) cells. For simplicity, assume that each cell has only one BS equipped with $N_t\geq 1$ transmit antennas, providing service to a set of ${\cal K}:=\{1,\ldots, K\}$ single-antenna
user equipments (UEs); see Fig. \ref{model}. The $I$ BSs operate over a common frequency band, and each communicates with its associated UEs using transmit beamforming.
Note that our approach can be readily extended to the scenarios where one BS serves a different number of UEs and/or one UE is served by multiple BSs.

\begin{figure}[th]
\centering
\includegraphics[width=0.6\textwidth]{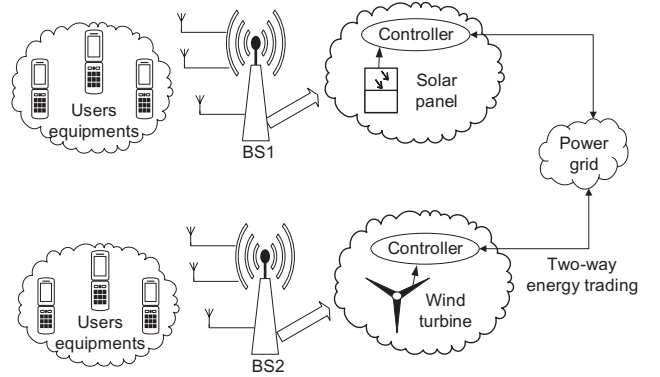}
\caption{A smart-grid powered coordinated multicell downlink system.}
\label{model}
\end{figure}

Assume that the multicell downlink system is powered by a smart microgrid,
where each BS is equipped with one or more local energy harvesting devices (solar panels and/or wind turbines).
Upon energy surplus or deficit in the microgrid, the BSs can perform two-way (i.e., buying/selling) energy trading with the main grid.
A controller is installed at each BS, to collect information of electricity prices
and to coordinate energy transaction activities with main grid.

\subsection{Downlink Transmission Model}\label{sec:downlink}

Let $\text{BS}_i$ denote the $i$th BS, and $\text{UE}_{ik}$ denote the $k$th user served by $\text{BS}_i$, for $i \in {\cal I}$ and $k \in {\cal K}$.
With $s_{ik}(t)$ denoting the information-bearing symbol for $\text{UE}_{ik}$ and $\mathbf{w}_{ik} \in \mathbb{C}^{N_t \times 1}$ the associated beamforming vector, the transmit signal at $\text{BS}_i$ is given by
\[
    \mathbf{x}_{i}(t) =\sum_{k=1}^K {\mathbf{w}_{ik} s_{ik}(t)}, \quad \forall i
\]

Let $\mathbf{h}_{jik} \in \mathbb{C}^{N_t \times 1}$ denote the vector channel from $\text{BS}_j$ to $\text{UE}_{ik}$,
$\forall j, i \in \cal I$, $\forall k \in \cal K$.
The received signal at $\text{UE}_{ik}$ is then
\begin{align}
y_{ik}(t) &= \sum_{j=1}^I{\mathbf{h}_{jik}^H \mathbf{x}_j(t)} + z_{ik}(t)
= \mathbf{h}_{iik}^H \mathbf{w}_{ik} s_{ik}(t) + \sum_{l \neq k}^K {\mathbf{h}_{iik}^H \mathbf{w}_{il} s_{il}(t)} 
+ \sum_{j \neq i}^I {\sum_{l=1}^K {\mathbf{h}_{jik}^H \mathbf{w}_{jl} s_{jl}(t)}} + z_{ik}(t) \label{eq.yk}
\end{align}
where $z_{ik}(t)$ is a circularly symmetric complex Gaussian noise with zero mean and variance $\sigma_{ik}^2$.
Clearly, the first term in (\ref{eq.yk}) is the signal of interest, while the second and third terms are intra-BS and inter-BS interferences, respectively.

Assume that $s_{ik}(t)$ is statistically independent, with zero mean and unit variance, e.g., $\mathbb{E}\{|s_{ik}(t)|^2\}=1$, $\forall i, k$,
and that each UE employs single-user detection. Further define the downlink channel covariance matrices
$\mathbf{R}_{jik}:= \mathbb{E}[\mathbf{h}_{jik}\mathbf{h}_{jik}^H]$, $\forall j, i,k$,
where $\mathbb{E}[\cdot]$ denotes the average over the time-fluctuating fading channel.
Based on the signal model (\ref{eq.yk}),
the long-term signal-to-interference-plus-noise-ratio (SINR) at $\text{UE}_{ik}$ can be expressed as
\cite{schubert04, huang10, alex10}
\begin{align}\label{sinr}
    \text{SINR}_{ik}(\{\mathbf{w}_{ik}\}) 
    = \frac{\mathbf{w}_{ik}^H\mathbf{R}_{iik}\mathbf{w}_{ik}}{\sum_{l\neq k} (\mathbf{w}_{il}^H\mathbf{R}_{iik}\mathbf{w}_{il}) +
    \sum_{j \neq i}{\sum_{l}{\mathbf{w}_{jl}^H\mathbf{R}_{jik}\mathbf{w}_{jl}}} +
    \sigma_{ik}^2}~.
\end{align}
To guarantee the user QoS, it is required that
\begin{equation}\label{eq.sinr}
    {\text{SINR}}_{ik} (\{\mathbf{w}_{ik}\}) \geq \gamma_{ik}, \quad \forall i, k
\end{equation}
where $\gamma_{ik}$ denotes the target SINR value per $\text{UE}_{ik}$.

\subsection{RES and Two-Way Energy Trading}\label{sec:sg}

Given the beamforming vectors $\{\mathbf{w}_{ik}\}$, the transmit power of each $\text{BS}_i$ is given by
\begin{equation}\label{eq.pipi}
    P_{i} = \sum_{k=1}^K \mathbf{w}_{ik}^H \mathbf{w}_{ik}.
\end{equation}
In conventional cellular networks, the BSs can only buy energy from the grid to supply the needed $P_{i}$ for data transmission.
Powered by a smart microgrid, the BSs in our system can harvest RES. They need only to buy energy when the amount of harvested energy is insufficient. Furthermore, with two-way energy trading, BSs can even sell surplus energy to the main grid to reduce transaction cost.

Let $e_i$ denote the (random) energy harvested by $\text{BS}_i$, which is generated according to a stationary random process with unknown distribution. 
For convenience, the energy harvesting interval is normalized to unity; thus, the terms ``energy'' and ``power'' will be used interchangeably throughout the paper.

For $\text{BS}_i$, the energy shortfall or surplus is $[P_i - e_i]^+$ or $[e_i - P_i]^+$, where $[a]^+:= \max\{a,0\}$. Clearly, both the
shortage and surplus energy amounts are non-negative, and we have at most one of them be positive. Suppose that the energy can be purchased from the grid at price $a_i$, while the energy is sold to the grid at price $b_i$ per $\text{BS}_i$.
To fully harness the capability of RES, the transaction prices may fluctuate with the amount of harvested energy; e.g.,
when grid-wise RES amount increases, the buying and selling prices may go down, encouraging more energy consumption from the end users.
Hence, the transaction prices obey a certain distribution that could be correlated to the random process of RES.
Note that we shall always have $a_i \geq b_i$,\footnote{For American electricity markets, a single pricing mechanism is used where $a_i = b_i, \forall i$ holds in most of the scenarios. This is a special case of the stated pricing situation, which in fact facilitates the calculation of cost in \eqref{eq.cost} since the absolute value function would vanish \cite{Yu15}.} to avoid meaningless buy-and-sell activities of the BS for profit.
With the random state $s_i :=\{a_i, b_i, e_i\}$, $\forall i$, the net transaction cost per $\text{BS}_i$ is thus given by
\begin{align}\label{eq.cost}
\hat{f}_i(P_i, s_i) &= a_i \cdot [P_i - e_i]^+ - b_i \cdot [e_i - P_i]^+ 
= \alpha_i \cdot |P_i - e_i| + \beta_i \cdot (P_i - e_i)
\end{align}
where $\alpha_i := \frac{a_i-b_i}{2} \geq 0 $ and $\beta_i := \frac{a_i+b_i}{2} \geq 0$.

\subsection{Two Assumptions}

For the RES powered multicell downlink system, the following operational conditions are assumed.
\begin{enumerate}
\item[{\bf as-1)}] The downlink channel covariance matrices $\mathbf{R}_{ijk}$ remain unchanged, $\forall i,j,k$; and each $\text{BS}_i$ has the knowledge of local CSI $\{\mathbf{R}_{ijk}\; \forall j,k\}$ available.

\item[{\bf as-2)}] The buying price $a_i$, selling price $b_i$, and RES amount $e_i$ are unknown; yet, each $\text{BS}_i$ has a local database where a large number of $s_i:=\{a_i,b_i,e_i\}$ realizations are collected (in the past).
\end{enumerate}

The invariant $\mathbf{R}_{ijk}$ assumption in as-1) holds as long as the downlink channels $\mathbf{h}_{ijk}$ are (wide-sense) stationary over the resource scheduling interval of interest. As $\mathbf{R}_{ijk}$ can remain unchange over a relatively long period, it is also reasonable to assume that each $\text{BS}_i$ can have the local $\{\mathbf{R}_{ijk}\; \forall j,k\}$ available through e.g., effective channel-covariance estimation and feedback schemes.
By as-2), we essentially look for an efficient ahead-of-time scheduling in the presence of uncertain buying, selling prices and RES amounts.
In the current big data era, however, a large number of $s_i$ realizations (i.e., collected from past measurements) can be stored in the database per $\text{BS}_i$. Hence, a data-driven approach can be developed to deal with the stochasticity of $s_i$.

Note that the assumption of uncertain electricity prices and RES amounts in as-2) is consistent with the typical smart-grid scenarios. As the power grid is shouldering heavier loads during peak hours everyday, a dynamic electricity pricing mechanism needs to be employed to balance power demands between peak and off-peak hours to maintain grid-wide stability. With high-penetration RES and two-way energy trading, the amounts of harvested energy can be considered in deciding transaction prices, in order to maintain a stable power output of the main grid. In the electricity markets, an ahead-of-time energy scheduling and real-time dispatch policy can be then adopted, where the RES amounts and real-time electricity prices are unknown \cite{Yu15}. In the presence of such uncertainties, an efficient method to minimize the average transaction cost, as well as to control the risk of very high cost, is critical for the smart-grid powered BSs. 

\section{Convex Problem Formulation}\label{sec:cvxform}

Based on the models of Section II, we next formulate our CMBF design problem. Integrating RES into the CMBF design requires risk-cognizant dispatch of resources to account for the stochastic availability of renewables. To this end, our idea is to capitalize on the novel notion of CVaR.

\subsection{CVaR based Energy Transaction Cost}\label{sec:cvarcost}

CVaR has been widely used in various real-world applications, especially in the finance area, to account not only for the expected cost of the resource allocation actions, but also for their ``risks'' \cite{rocka02, rocka00, quaranta08, Yu15}. In the present context, recall that the transaction cost $\hat{f}_i(P_i, s_i)$ in (\ref{eq.cost}) is a function associated with the decision variable $P_i$ and the random state $s_i:=\{a_i,b_i,e_i\}$. Assume that the transaction prices $a_i$ and $b_i$, and the RES amount $e_i$ are generated according to some stationary random processes with the joint probability density function $p(s_i)$.
The probability of $\hat{f}_i(P_i, s_i)$ not exceeding a threshold $\eta_i$ is then given by the right-continuous
cumulative distribution function 
\begin{equation}
\Psi(P_i,\eta_i) = \int\limits_{\hat{f}_i(P_i,s_i)\le \eta_i} p(s_i)\, ~\text{d} s_i.
\end{equation}

Given a prescribed confidence level $\theta \in (0, 1)$, we can define the $\theta$-value-at-risk (VaR) as the generalized inverse of $\Psi$, i.e.,
\begin{equation}
\eta_\theta(P_i) := \min\{\eta_i \in\mathbb{R}~|~\Psi(P_i,\eta_i) \ge \theta\}. \label{eq:VaR}
\end{equation}
Clearly, $\theta$-VaR is essentially the $\theta$-quantile of the random $\Psi(P_i,\eta_i)$. Since $\Psi$ is non-decreasing in $\eta_i$, $\eta_\theta(P_i)$ comes out as the lower endpoint of the solution interval satisfying $\Psi(P_i,\eta_i) \ge \theta$.

Based on $\eta_\theta(P_i)$, we can further define the $\theta$-CVaR ($\Psi_{\theta}$) as the mean of the $\theta$-tail distribution of $\hat{f}_i(P_i,s_i)$, which is given by
\begin{equation}\label{def.cvar}
\Psi_{\theta}(P_i,\eta_i) :=
\left\{\begin{array}{cc}
0, &\mbox{if}~\eta_i < \eta_\theta(P_i)\\
\frac{\Psi(P_i,\eta_i)-\theta}{1-\theta},  &\mbox{if}~\eta_i \geq \eta_\theta(P_i)
\end{array}\right.
\end{equation}
Truncated and re-scaled from $\Psi$, function $\Psi_{\theta}$ is nondecreasing, right-continuous, and in fact a (conditional) distribution function. As a result, $\theta$-CVaR is the expected cost in the worst $100(1-\theta)\%$ scenarios. Clearly, as $\theta \rightarrow 0$, the $\theta$-CVaR becomes the (unconditional) expected transaction cost. On the other hand, as $\theta \rightarrow 1$, the $\theta$-CVaR approaches to the (conservative) worst-case transaction cost. The value of $\theta$ is commonly chosen as no more than $0.99$; use of the corresponding $\theta$-CVaR as cost function could lead to a both efficient and robust resource allocation action.

It was shown that the $\theta$-CVaR can be also obtained as the optimal value of the following optimization problem \cite{rocka02}
\begin{equation}
\phi_\theta(P_i) := \min_{\eta_i \in \mathbb{R}} \left\{\eta_i +\frac{1}{1-\theta}\mathbb{E}_{s_i}\{[\hat{f}_i(P_i, s_i)-\eta_i]^{+}\}\right\}.\label{eq:CVaR}
\end{equation}
Define $F_i(P_i, \eta_i):=\eta_i +\frac{1}{1-\theta}\mathbb{E}_{s_i}\{[\hat{f}_i(P_i, s_i)-\eta_i]^{+}\}$. As with \cite[Proposition 1]{Yu15}, we have:
\begin{lemma}
The function $F_i(P_i, \eta_i)$ is jointly convex in $(P_i, \eta_i)$.
\end{lemma}
\begin{proof}
Given that $a_i \geq b_i$, it follows by convexity of the absolute value function that $\hat{f}_i(P_i, s_i)$ is convex in $P_i$. Due to the convexity-preserving operators of projection and expectation, the lemma readily follows \cite[Sec.~3.2]{convex}.
\end{proof}

\begin{remark}\textit{(Properties of the cost function)}
The condition $a_i \geq b_i$, $\forall i \in \cal{I}$, is sufficient but not necessary to guarantee that the function $\hat{f}_i(P_i, s_i)$ is convex. The objective function can take forms other than the conditional expected transaction cost function. Our proposed method is still applicable as long as the objective function is convex with respect to its effective domain. This can even include some cases when the selling prices exceed the purchase prices. 
\end{remark}

The convexity of $F_i(P_i, \eta_i)$ facilitates the next convex CMBF problem formulation.


\subsection{CMBF Problem}\label{sec:cmbfp}

Adopting the CVaR-based cost function, we pursue the CMBF design that minimizes system-wide energy transaction cost subject to (s.t.) user SINR constraints. By (\ref{eq:CVaR}), minimizing CVaR $\phi_\theta(P_i)$ with respect to $P_i$ is equivalent to minimizing $F_i(P_i, \eta_i)$ over $(P_i, \eta_i)$.
Then mathematically, our problem can be formulated as
\begin{subequations}~\label{eq.task}
\begin{align}
\min_{\{P_i, \eta_i, \mathbf{w}_{ik}\}} &  \sum_{i=1}^{I} F_i(P_i, \eta_i) \label{eq.prob} \\
\mathrm{s.t.} &~~P_{i} = \sum_{k=1}^K \mathbf{w}_{ik}^H \mathbf{w}_{ik} , \quad \forall i  \label{eq.w}\\
&~~{\text{SINR}}_{ik} (\{\mathbf{w}_{ik}\}) \geq \gamma_{ik}, \quad \forall i,\forall k . \label{eq.sinrm}
\end{align}
\end{subequations}


By Lemma 1, the objective function \eqref{eq.prob} is convex. We next rely on the popular semidefinite program (SDP)
relaxation technique to convexify the non-convex constraints (\ref{eq.sinrm}). To this end, we rewrite \eqref{eq.sinrm} into
\begin{align}
   \frac{1}{\gamma_{ik}}\mathbf{w}_{ik}^H\mathbf{R}_{iik}\mathbf{w}_{ik} - \sum_{l\neq k}\mathbf{w}_{il}^H\mathbf{R}_{iik}\mathbf{w}_{il} 
   \geq \sum_{j\neq i}^I
\sum_{l=1}^{K}{\mathbf{w}_{jl}^H\mathbf{R}_{jik}\mathbf{w}_{jl}}
 +\sigma_{ik}^2. \label{eq.sinrm2}
\end{align}
Let $\mathbf{W}_{ik} := \mathbf{w}_{ik} \mathbf{w}_{ik}^H$. It clearly holds that $\mathbf{W}_{ik} \succeq \mathbf{0}$, and $\text{rank}(\mathbf{W}_{ik}) = 1$, $\forall i, k$. Dropping the latter rank constraints, the SINR constraints (\ref{eq.sinrm2}) can be relaxed to the convex SDP constraints, and the problem (\ref{eq.task}) becomes
\begin{subequations}~\label{tasknew}
\begin{align}
& \min_{\{P_i, \eta_i, \mathbf{W}_{ik}\}} \sum_{i=1}^{I} F_i(P_i, \eta_i) \label{tasknew1} \\
\mathrm{s.t.} &~~P_i=\sum_{k=1}^{K}\text{tr}(\mathbf{W}_{ik}) , ~~\forall i \label{tasknew0}\\
&~~\frac{1}{\gamma_{ik}}\text{tr}(\mathbf{R}_{iik}\mathbf{W}_{ik})-\sum_{l\neq k}\text{tr}(\mathbf{R}_{iik}\mathbf{W}_{il}) 
\geq \sum_{j\neq i}^I \sum_{l=1}^{K}
\text{tr}(\mathbf{R}_{jik}\mathbf{W}_{jl})+\sigma_{ik}^2, ~~\forall i,k \label{tasknew3}\\
&~~\mathbf{W}_{ik} \succeq \mathbf{0},~~ \forall i,k. ~\label{tasknew4}
\end{align}
\end{subequations}
where $\text{tr}(\cdot)$ denotes the trace operator. Problem (\ref{tasknew}) is a convex program that can be efficiently solved by interior-point methods in polynomial time. With $\mathbf{W}_{ik}^*$, $\forall i,k$, denoting the optimal beamforming matrices for (\ref{tasknew}), we can show that:
\begin{lemma}
Under either of the following two conditions: i) $\text{rank}(\mathbf{R}_{ijk})=1$, $\forall i, j, k$, or ii) $I \leq 2$, we can always have $\text{rank}(\mathbf{W}_{ik}^*) = 1$, $\forall i, k$, for problem (\ref{tasknew}).
\end{lemma}

\begin{proof}
Consider the two conditions one by one.

c-1): $\mathbf{R}_{ijk}$ is rank-one, i.e., $\mathbf{R}_{ijk}=\mathbf{h}_{ijk}\mathbf{h}_{ijk}^H$, $\forall i,j,k$.
In this case, \eqref{eq.task} can be recast into
\begin{subequations}~\label{taskproof1}
\begin{align}
& \min_{\{P_i, \eta_i, \mathbf{w}_{ik}\}} \sum_{i=1}^{I} F_i(P_i, \eta_i) \label{taskproof10} \\
\mathrm{s.t.} &~~\sum_{k=1}^{K}\text{tr}(\mathbf{w}_{ik}\mathbf{w}_{ik}^H)\leq P_i , ~~\forall i \label{taskproof11}\\
&\frac{1}{\gamma_{ik}}|\mathbf{h}_{iik}^H\mathbf{w}_{ik}|^2 - \sum_{l\neq k}|\mathbf{h}_{iik}^H\mathbf{w}_{il}|^2 
   \geq \sum_{j\neq i}^I\sum_{l=1}^{K}{|\mathbf{h}_{jik}^H\mathbf{w}_{jl}}|^2+\sigma_{ik}^2, ~~\forall i,k. \label{taskproof12}
\end{align}
\end{subequations}
Note that we are allowed to replace ``='' with ``$\leq$'' in \eqref{taskproof11} since $F_i(P_i, \eta_i)$ is monotonically increasing in $P_i$.
Now, \eqref{taskproof11} is in fact a second-order cone (SOC) constraint, $\forall i$, since it is identical to
$\sum_{k} \| \mathbf{w}_{ik}\|^2 \leq \frac{(1+P_i)^2}{4}-\frac{(1-P_i)^2}{4}$,
which is in turn identical to $\sqrt{\sum_{k} \| \mathbf{w}_{ik}\|^2 + (\frac{1-P_i}{2})^2} \leq \frac{1+P_i}{2}$.
The constraint \eqref{taskproof12} is also an SOC constraint as it is equivalent to
\begin{equation}\label{realpart}
\frac{\text{Re}(\mathbf{h}_{iik}^H\mathbf{w}_{ik})}{\gamma_{ik}}\geq
\sqrt{\sum_{l\neq k}|\mathbf{h}_{iik}^H\mathbf{w}_{il}|^2 + \sum_{j\neq i}^I\sum_{l=1}^{K}{|\mathbf{h}_{jik}^H\mathbf{w}_{jl}}|^2+\sigma_{ik}^2 }~,
\end{equation}
where $\text{Re}(\cdot)$ denotes the real part of a complex variable.
To see the equivalence between \eqref{taskproof12} and \eqref{realpart}, one can observe that for any $\mathbf{w}_{ik}$ satisfying \eqref{realpart},
a phase rotation $\mathbf{\tilde{w}}_{ik}:= \mathbf{w}_{ik}\cdot e^{-j(\mathbf{h}_{iik}^H\mathbf{w}_{ik})}$ is feasible for \eqref{taskproof12}.
The two are equivalent as arbitrary phase rotation for beamforming vectors would not affect both the transmit power and the quadratic constraints of interest.
Hence, the original problem \eqref{eq.task} can be reformulated as the following convex second-order cone program (SOCP):
\begin{subequations}~\label{taskproof2}
\begin{align}
& \min_{\{P_i, \eta_i, \mathbf{w}_{ik}\}} \sum_{i=1}^{I} F_i(P_i, \eta_i) \label{taskproof20} \\
\mathrm{s.t.} &~~\sum_{k=1}^{K}\| \mathbf{w}_{ik}\|^2 \leq P_i , ~~\forall i \label{taskproof21}\\
&\sqrt{\sum_{l\neq k}|\mathbf{h}_{iik}^H\mathbf{w}_{il}|^2 + \sum_{j\neq i}^I\sum_{l=1}^{K}{|\mathbf{h}_{jik}^H\mathbf{w}_{jl}}|^2+\sigma_{ik}^2 } \leq \frac{1}{\gamma_{ik}}\text{Re}(\mathbf{h}_{iik}^H\mathbf{w}_{ik}) ,~\forall i,k. \label{taskproof22}
\end{align}
\end{subequations}
The optimal transmit beamforming vectors $\{\mathbf{w}_{ik}^*\}_{i,k}$ can be directly obtained by solving \eqref{taskproof2}.
It then readily follows that there always exists the rank-one optimal solution $\mathbf{W}_{ik}^*=\mathbf{w}_{ik}^*{\mathbf{w}_{ik}^*}^H$,
$\forall i,k$, for \eqref{tasknew}.

c-2): $I\leq 2$, and $\mathbf{R}_{ijk}$ can be of any rank.
Suppose that we have solved and obtained the optimal power values $\{P_i^*, \forall i\}$ for \eqref{tasknew}.
Consider the following optimization problem:
\begin{subequations}~\label{taskproof3}
\begin{align}
& \min_{\{\mathbf{W}_{ik}\}} 0 \label{taskproof30} \\
\mathrm{s.t.} &~~\sum_{k=1}^{K}\text{tr}(\mathbf{W}_{ik}) = P_i^* , ~~\forall i \label{taskproof31}\\
&~~\frac{1}{\gamma_{ik}}\text{tr}(\mathbf{R}_{iik}\mathbf{W}_{ik}) \geq \sum_{l\neq k}\text{tr}(\mathbf{R}_{iik}\mathbf{W}_{il}) 
+\sum_{j\neq i}^I \sum_{l=1}^{K} \text{tr}(\mathbf{R}_{jik}\mathbf{W}_{jl})+\sigma_{ik}^2, ~~\forall i,k \label{taskproof32}\\
&~~\mathbf{W}_{ik} \succeq \mathbf{0},~~ \forall i,k. \label{taskproof33}
\end{align}
\end{subequations}
Problem (\ref{taskproof3}) is an SDP with $IK$ variables and $IK+I$ constraints.
Let $\mathbf{\tilde{W}}_{ik}^*, \forall i, k$ denote the optimal solution for \eqref{taskproof3}.
By \cite[Theorem 3.2]{huang10}, there always exists $\{\mathbf{\tilde{W}}_{ik}^*\}_{i,k}$ with
\begin{equation}\label{rankik}
\sum_{i=1}^I \sum_{k=1}^K \text{rank}^2(\mathbf{\tilde{W}}_{ik}^*) \leq IK+I.
\end{equation}
It is also clear that all matrices $\{\mathbf{\tilde{W}}_{ik}^*\}_{i,k}$ are non-zero matrices, i.e., $\text{rank}(\mathbf{\tilde{W}}_{ik}^*)\geq 1$, $\forall i,k$, otherwise the corresponding SINR constraints will be violated. Together with \eqref{rankik}, we then readily have: $\text{rank}(\mathbf{\tilde{W}}_{ik}^*)= 1$, $\forall i,k$, when $I \leq 2$. It is easy to show that such rank-one $\{\mathbf{\tilde{W}}_{ik}^*\}_{i,k}$ are also feasible and optimal for problem \eqref{tasknew}; the lemma readily follows.
\end{proof}

\begin{remark}\textit{(Tightness of SDP relaxation)}
Lemma 2 implies that the SDP relaxation (\ref{tasknew}) is tight when i) $\text{rank}(\mathbf{R}_{ijk})=1$, $\forall i, j, k$, or ii) $I \leq 2$. Case i) in fact corresponds to the time-invariant channel case considered in our conference version \cite{conf}. Case ii) holds when there are at most two BSs involved in the coordinated beamforming. Although the tightness of the SDR (\ref{tasknew}) could not be rigorously proved for other general cases, our extensive numerical results indicate that this actually always holds.\footnote{Proof for the tightness of such an SDR for general cases will be an interesting direction to pursue in our future works.} With the rank-one matrices $\{\mathbf{W}_{ik}^*\}$, the optimal CMBF solution $\{\mathbf{w}_{ik}^*\}$ for the original (\ref{eq.task}) can be exactly obtained by eigen-decomposition,
i.e., $\mathbf{W}_{ik}^* = \mathbf{w}_{ik}^* {\mathbf{w}_{ik}^*}^H$, $\forall i, k$.
\end{remark}
\begin{remark}\textit{(Approximate optimality of $\{\mathbf{w}_{ik}\}$)}
In more general cases where the optimal solutions $\{\mathbf{W}_{ik}^*\}$ for the SDP relaxed problem \eqref{tasknew}
are not rank-one, the suboptimal solutions $\{\mathbf{w}_{ik}\}$ for the original problem \eqref{eq.task}
can be efficiently obtained via a randomization and scaling algorithm proposed in \cite{nich06}.
Essentially we calculate the eigen-decomposition of $\mathbf{W}_{ik}^* = \mathbf{U}_{ik}\mathbf{\Sigma}_{ik}\mathbf{U}^H_{ik},
\forall i,k$, then design $\mathbf{w}_{ik}$ by a vector of carefully chosen random variables such that
$\mathbf{w}_{ik} = \mathbf{U}_{ik}\mathbf{\Sigma}^{1/2}_{ik}\mathbf{v}_{ik}$,
and $\mathbb{E}[\mathbf{w}_{ik}\mathbf{w}_{ik}^H]=\mathbf{W}_{ik}^*$.
Scale $\mathbf{w}_{ik}$ if it violates constraints \eqref{eq.w} and \eqref{eq.sinrm},
so that all $\{\mathbf{w}_{ik}\}$ satisfy the constraints.
The ``best'' beamformers for the original problem  \eqref{eq.task} are the ones
that require the smallest scaling, which are then outputs as the approximate solutions.
Interested readers can refer to \cite{nich06} for a detailed discussion.
\end{remark}

\section{Distributed CMBF via Stochastic ADMM}\label{sec:distcmbf}

Directly solving the SDP problem (\ref{tasknew}) calls for a central controller which has the global CSI. For the multicell downlink system, it is certainly desirable to obtain the CMBF solution in a decentralized manner using only local CSI at each BS per as-1).
In the present context, development of such a distributed solver also needs to take into account the stochasticity of RES.
To this end, we resort to a {\em stochastic} ADMM approach.

\subsection{Review of Stochastic ADMM}

To illustrate the idea of stochastic ADMM, let us consider the following separable convex minimization problem
with linear equality constraints:
\begin{subequations}
\label{eq:genADMM}
\begin{align}
&\min_{\mathbf{x}\in \mathcal{X}, \mathbf{z}\in \mathcal{Z}}
 \mathbb{E}_{\boldsymbol{\xi}}\{f(\mathbf{x}, \boldsymbol{\xi})\} + g(\mathbf{z}) \label{genADMM1}\\
&\quad\mathrm{subject~to:}~ \mathbf{B}\mathbf{z} = \mathbf{A}\mathbf{x}
\label{eq:genADMM-lin}
\end{align}
\end{subequations}
where $\boldsymbol{\xi}$ is a random vector, obeying a fixed but unknown distribution. 
Since the probability distribution function of $\boldsymbol{\xi}$ is unknown,
the deterministic ADMM principle cannot be directly applied to solve (\ref{eq:genADMM}).
Suppose that a sequence of i.i.d. observations for the random vector $\boldsymbol{\xi}$ can be drawn.
To apply the deterministic ADMM for \eqref{eq:genADMM}, we need to approximate the first term in \eqref{genADMM1}
by its empirical expectation via Monte Carlo sampling \cite{Yu15},
which requires a visit to all the samples per iteration.
This can be slow and computationally expensive in the current big data era due to data proliferation.
Hence, generalizing the classical and the linearized ADMM, a {\em stochastic} ADMM approach was then proposed in \cite{hua13}.
Specifically, we define an {\em approximated} augmented Lagrangian function
\begin{align}\label{eq.augprox}
&\hat{\cal L}_{\rho, m}(\mathbf{x}, \mathbf{z}, \boldsymbol{\lambda}):= f(\mathbf{x}_m, \boldsymbol{\xi}_{m+1}) 
+ \mathbf{x}^T f'(\mathbf{x}_m, \boldsymbol{\xi}_{m+1}) + g(\mathbf{z})\notag\\
&\quad-\boldsymbol{\lambda}^T(\mathbf{B}\mathbf{z} -\mathbf{A}\mathbf{x}) +\frac {\rho}{2} \|\mathbf{B}\mathbf{z} -\mathbf{A}\mathbf{x}\|^2
+ \frac{\|\mathbf{x} - \mathbf{x}_m\|^2}{2\zeta_{m+1}}
\end{align}
where $m$ is the iteration index, $\boldsymbol{\lambda}$ collects the Lagrangian multipliers associated with the equality constraint, and $\rho$ is a pre-defined penalty parameter controlling the violation of primal feasibility, as with the classic ADMM. Yet, we replace $\mathbb{E}_{\boldsymbol{\xi}}\{f(\mathbf{x}, \boldsymbol{\xi})\}$ with a first-order approximation of $f(\mathbf{x}, \boldsymbol{\xi}_{m+1})$ at  $\mathbf{x}_m$: $f(\mathbf{x}_m, \boldsymbol{\xi}_{m+1}) + \mathbf{x}^T f'(\mathbf{x}_m, \boldsymbol{\xi}_{m+1})$,
in the same flavor of the stochastic mirror descent.
Similar to the linearized ADMM, we also add an $l_2$-norm prox-function $\|\mathbf{x} - \mathbf{x}_m\|^2$
but scale it by a time-varying stepsize $\zeta_{m+1}$,
which is usually set as ${O}(1/\sqrt{m})$ to ensure fast convergence \cite{hua13}.
The stochastic ADMM procedures are summarized in Algorithm 1.

For the stochastic ADMM algorithm, in each iteration $\mathbf{x}$ is updated based on a single (random) sample; hence, the update costs little time and resources. On the other hand, it was established that this method can approach the globally optimal solution to (\ref{eq:genADMM}) in expectation (or in probability) with good rates of convergence \cite{hua13}.

\begin{algorithm}[t]
\caption{Stochastic ADMM Approach}
\label{algo:ORC}
\begin{algorithmic}[1]
\State {\bf Initialize} $\mathbf{x}_0$, $\mathbf{z}_0$, and set $\boldsymbol{\lambda}_0 = \bf0$.
\For {$m$ = 0, 1, 2, ...}
\State $\mathbf{x}_{m+1} = \argmin\limits_{\mathbf{x}\in \mathcal{X}} \hat{\cal L}_{\rho, m}(\mathbf{x}, \mathbf{z}_m, \boldsymbol{\lambda}_m)$.
\State $\mathbf{z}_{m+1} = \argmin\limits_{\mathbf{z}\in \mathcal{Z}} \hat{\cal L}_{\rho, m}(\mathbf{x}_{m+1}, \mathbf{z}, \boldsymbol{\lambda}_m)$.
\State $\boldsymbol{\lambda}_{m+1} = \boldsymbol{\lambda}_{m} - \rho (\mathbf{B}\mathbf{z}_{m+1} - \mathbf{A}\mathbf{x}_{m+1})$.
\EndFor
\end{algorithmic}
\end{algorithm}

\setcounter{TempEqCnt}{\value{equation}}
\setcounter{equation}{25}
\begin{figure*}[b]
\centering
\hrulefill
\begin{align}
\label{eq.augprox1}
\hat{\cal L}_{\rho, m}(\mathbf{x}, \mathbf{z}, \boldsymbol{\lambda})=\sum_i \Big\{& f_i(P_i(m),\eta_i(m))
+P_i \frac{\partial f_i(m)}{\partial P_i}
+\eta_i \frac{\partial f_i(m)}{\partial \eta_i}
-\bm\lambda_i^T(\mathbf{B}_i\mathbf{\bar{q}}-\mathbf{q}_i)+\frac{\rho}{2}\|\mathbf{B}_i\mathbf{\bar{q}}-\mathbf{q}_i\|^2 \notag\\
&+(\|P_i-P_i(m)\|^2 + \|\mathbf{q}_i- \mathbf{q}_i(m)\|^2 + \|\eta_i-\eta_i(m)\|^2)/2\zeta(m+1)
\Big\}
\end{align}
\setcounter{equation}{\value{TempEqCnt}}
\end{figure*}

\subsection{Problem Reformulation}\label{sec:reform}

We next reformulate \eqref{tasknew} such that the stochastic ADMM procedures in Algorithm 1 can be applied to obtain the optimal CMBF solution in a distributed fashion.

To this end, we introduce the auxiliary variables $q_{jik}:=\sum_{l=1}^{K}
\text{tr}(\mathbf{R}_{jik}\mathbf{W}_{jl})$, $\forall j \neq i$. Clearly, $q_{jik}$ is the inter-BS interference power from $\text{BS}_j$ to $\text{UE}_{ik}$. Further introduce another set of auxiliary variables $Q_{ik} := \sum_{j \neq i} q_{jik}$, $\forall i,k$, which represent the total inter-BS interference power from the neighboring BSs to $\text{UE}_{ik}$.

Using $\{q_{jik}\}$ and $\{Q_{ik}\}$, we can rewrite (\ref{tasknew}) into
\begin{subequations}~\label{sadmm}
\begin{align}
& \min_{\{P_i, \eta_i, \mathbf{W}_{ik}, Q_{ik}, q_{jik}\}} \sum_{i=1}^{I} F_i(P_i, \eta_i) \label{sadmm1} \\
& \mathrm{s.t.} ~~P_i = \sum_{k=1}^{K}\text{tr}(\mathbf{W}_{ik}) , ~~\forall i \label{sadmm2}\\
&~~ q_{jik} = \sum_{l=1}^{K} \text{tr}(\mathbf{R}_{jik}\mathbf{W}_{jl}), ~~\forall j,i,k, \& j \neq i \label{sadmm3} \\
&~~\frac{1}{\gamma_{ik}}\text{tr}(\mathbf{R}_{iik}\mathbf{W}_{ik})-\sum_{l\neq k}\text{tr}(\mathbf{R}_{iik}\mathbf{W}_{il}) \geq Q_{ik}+\sigma_{ik}^2, ~~\forall i,k \label{sadmm4}\\
&~~\mathbf{W}_{ik} \succeq \mathbf{0},~~ \forall i,k  \label{sadmm5} \\
&~~ Q_{ik} = \sum_{j \neq i} q_{jik}, ~~ \forall i,k  \label{sadmm6}
\end{align}
\end{subequations}
It can be easily observed from the SINR constraints \eqref{sadmm4} that each $\text{UE}_{ik}$ concerns only the total inter-BS interference power $Q_{ik}$ rather than the individual inter-BS interference powers $\{q_{jik}\}$. Also note that we can interchange the subindices $j$ and $i$ in \eqref{sadmm3} without changing the problem. As a result, we can decompose the constraints \eqref{sadmm2}--\eqref{sadmm5} into $I$ independent convex sets: $\forall i$,
\begin{align}\label{cvxset}
 {\cal C}_i = &\Big\{\big  (P_i, \{Q_{ik}\}_{k}, \{q_{ijk}\}_{j,k}\big) \big| \notag\\
&~~~~P_i = \sum_{k=1}^{K}\text{tr}(\mathbf{W}_{ik}), ~q_{ijk} = \sum_{l=1}^{K} \text{tr}(\mathbf{R}_{ijk}\mathbf{W}_{il}), ~\forall j \neq i, \forall k, \notag \\
&~~~~ \frac{1}{\gamma_{ik}}\text{tr}(\mathbf{R}_{iik}\mathbf{W}_{ik})-\sum_{l\neq k}\text{tr}(\mathbf{R}_{iik}\mathbf{W}_{il}) \geq Q_{ik}+\sigma_{ik}^2, ~\forall k, \notag\\
&~~~~\mathbf{W}_{ik} \succeq \mathbf{0}, ~Q_{ik} \geq 0, ~\forall k \Big\}.
\end{align}
Note that we omit variables $\{\mathbf{W}_{ik}\}$ in ${\cal C}_i, \forall i$, as they can be seen as implicit optimization variables
in later formulation \eqref{eq.task2}.

Further define the following vectors:
\begin{subequations}\label{eq.q}
\begin{align}
&\mathbf{q}_i = \big[[Q_{i1},\ldots, Q_{iK}], [q_{i11},\ldots, q_{i1K}],\ldots, [q_{iI1},\ldots, q_{iIK}]\big]^T
 \in \mathbb{R}_{+}^{IK}, ~~\forall i \label{eq.q2} \\
&\mathbf{\bar{q}} = \big[[\bar{q}_{121},\ldots, \bar{q}_{12K}],\ldots, [\bar{q}_{I(I-1)1},\ldots, \bar{q}_{I(I-1)K}]\big]^T
\in \mathbb{R}^{I(I-1)K} \label{eq.q1}
\end{align}
\end{subequations}
where $\mathbf{q}_i$ collects variables $\{Q_{ik}\}_{k=1}^{K}$ and $\{q_{ijk}\}_{j,k}$ (with $j \neq i$) that are only relevant to $\text{BS}_i$, and $\mathbf{\bar{q}}$ collects a ``public'' copy of all the inter-BS interferences $\{\bar{q}_{ijk}\}_{i,j,k}$. Here, $\{\bar{q}_{ijk}\}_{i,j,k}$ are another set of new auxiliary variables, and these inter-BS interference terms should remain the same as the ``private'' ones; i.e., $\bar{q}_{ijk} = q_{ijk}$, $\forall i,j,k$. By the latter constraints, it is not difficult to work out the linear mapping matrix $\mathbf{B}_i \in \{0, 1\}^{IK \times I(I-1)K}$, such that $\mathbf{q}_i = \mathbf{B}_i\mathbf{\bar{q}}$,  $\forall i$.

With the introduction of the seemingly ``unnecessary'' auxiliary variables $\{Q_{ik}\}$, $\{q_{ijk}\}$, $\{\bar{q}_{ijk}\}$ as well as the sets $\{{\cal C}_i\}$, our problem can be now rewritten as
\begin{subequations}~\label{eq.task2}
\begin{align}
& \min_{\{P_i, \eta_i, Q_{ik}, q_{ijk}, \bar{q}_{ijk}\}} \sum_{i=1}^{I} F_i(P_i, \eta_i) \\
& \mathrm{s.t.} ~~ \big(P_i, \{Q_{ik}\}_{k}, \{q_{ijk}\}_{j,k}\big) \in {\cal C}_i, \\
&~~~~~~~\mathbf{q}_i = \mathbf{B}_i\mathbf{\bar{q}}, ~\forall i \label{eq.task2-2}
\end{align}
\end{subequations}
We next show that the reformulated (\ref{eq.task2}) is well suited for development of the desired distributed CMBF scheme.

\subsection{Distributed Solving Process via Stochastic ADMM}\label{sec:stoadm}

Define $f_i(P_i, \eta_i, s_i):=\eta_i +\frac{1}{1-\theta}[\hat{f}_i(P_i, s_i)-\eta_i]^{+}$. Clearly, we have $F_i(P_i, \eta_i)= \mathbb{E}_{s_i}\{f_i(P_i, \eta_i, s_i)\}$. Let $\mathbf{s}:=[s_1, \ldots, s_I]^T$. We can now identify the following correspondences between the variables in (\ref{eq.task2}) and \eqref{eq.augprox}:
\begin{align}
&\mathbf{x}\equiv [\mathbf{q}_1^T,\ldots, \mathbf{q}_I^T, P_1,\ldots, P_I, \eta_1,\ldots, \eta_I ]^T,~~\boldsymbol{\xi} \equiv \mathbf{s}, \notag\\
&\mathbf{z}\equiv \mathbf{\bar{q}}, ~~
\mathbf{A} \equiv \left(\begin{matrix}
\mathbf{I}_I & \mathbf{0}\\
\mathbf{0} & \mathbf{0}_{2I}
\end{matrix}\right),~~
\mathbf{B} \equiv [\mathbf{B}_1^T,\ldots,\mathbf{B}_I^T]^T, \notag \\
& \mathcal{X} \equiv {\bigcup}_i {\cal C}_i, ~~~ \mathcal{Z} \equiv \mathbb{R}^{I(I-1)K}; \notag
\end{align}
and the functions $f(\mathbf{x},\boldsymbol{\xi}) \equiv \sum_i f_i(P_i, \eta_i, s_i)$,
and $g(\mathbf{z}) \equiv 0$.
Note that as stated in as-2), a large number of realizations of $\{s_i\}_{i \in {\cal I}}$ are available at each BS.
The proposed stochastic ADMM scheme visits one sample per iteration.

We are ready to apply the principle of stochastic ADMM to solve (\ref{eq.task2}). With $\bm{\lambda}:=[\bm{\lambda}_1^T,\ldots,\bm{\lambda}_I^T]^T$ denoting the Lagrange multiplier vector associated with constraints \eqref{eq.task2-2},
the {\em approximated} augmented Lagrangian of \eqref{eq.task2} is given by \eqref{eq.augprox1} [cf. (\ref{eq.augprox})], where $\frac{\partial f_i(m)}{\partial P_i}$ and $\frac{\partial f_i(m)}{\partial \eta_i}$ denote the partial derivatives of $f_i(P_i, \eta_i, s_i)$ with respect to $P_i$ and $\eta_i$, evaluated at $(P_i(m), \eta_i(m), s_i(m+1))$, respectively. Their specific expressions are given by [cf. (\ref{eq.cost})]
\setcounter{equation}{23}
\begin{equation}\label{eq.fpi}
   \frac{\partial f_i(m)}{\partial P_i} =
    \begin{cases}
        \frac {a_i}{1-\theta},& \text{if } \hat{f}_i(m) \geq \eta_i(m) ~\&~ P_i(m)\geq e_i(m+1)\\
        \frac {b_i}{1-\theta},& \text{if } \hat{f}_i(m)\geq \eta_i(m) ~\&~ P_i(m)< e_i(m+1)\\
        0,& \text{if } \hat{f}_i(m)< \eta_i(m)
    \end{cases}
\end{equation}
\begin{equation}\label{eq.feta}
   \frac{\partial f_i(m)}{\partial \eta_i} =
    \begin{cases}
        \frac {-\theta}{1-\theta},& \text{if } \hat{f}_i(m)\geq \eta_i(m) \\
        1,& \text{if } \hat{f}_i(m)< \eta_i(m) \\
    \end{cases}
\end{equation}
where $\hat{f}_i(m)=\hat{f}_i(P_i(m), s_i(m+1))$.

\begin{algorithm}[t]
\caption{Distributed Resource Allocation Algorithm}
\label{algo:ORC}
\begin{algorithmic}[1]
\State {\bf Initialize} $\{P_i(0), \eta_i(0), \mathbf{q}_i(0), \bm{\lambda}_i(0)\}_{i=1}^{I}$, and $\mathbf{\bar{q}}(0)$ that are known to all BSs; choose a penalty parameter $\rho > 0$.
\For {$m$ = 0, 1, 2, ...}
\State Each BS solves the local beamforming design problems \eqref{eq.priup1} to obtain the local inter-BS interference vector $\mathbf{q}_i(m+1)$,
the energy consumption $P_i(m+1)$, and as a byproduct, the beamforming matrices $\{\mathbf{W}_{ik}(m+1)\}_{k=1}^K$. Each BS also solves the local problem \eqref{eq.priup3} to determine the threshold $\eta_i(m+1)$.
\State Each BS informs other BSs of its local inter-BS interference iterate $\mathbf{q}_i(m+1)$.
\State Each BS updates the public inter-BS interference vector $\mathbf{\bar{q}}(m+1)$ via \eqref{eq.priup4}.
\State Each BS updates the dual variable $\bm{\lambda}_i(m+1)$ via \eqref{eq.dual}.
\EndFor ~{until the predefined convergent criterion is met}.
\end{algorithmic}
\end{algorithm}

Based on the {\em approximated} augmented Lagrangian in \eqref{eq.augprox1}, problem \eqref{eq.task2} can be then tackled by solving the decomposed subproblems given by \eqref{eq.priup} and \eqref{eq.dual}, which update the primal and dual variables at $\text{BS}_i$, respectively, $\forall i$.
\setcounter{equation}{26}
\begin{subequations}\label{eq.priup}
\begin{align}
&\{P_i(m+1), \mathbf{q}_i(m+1)\} = \argmin\limits_{\big(P_i, \mathbf{q}_i\big) \in {\cal C}_i} P_i\frac{\partial f_i(m)}{\partial P_i}
+\bm\lambda_i^T(m)\mathbf{q}_i +\frac{\rho}{2}\|\mathbf{B}_i\mathbf{\bar{q}}(m)-\mathbf{q}_i\|^2 \notag\\
&\qquad \qquad \qquad \qquad \qquad \quad  + \frac{\|P_i-P_i(m)\|^2}{2\zeta(m+1)} +\frac{\|\mathbf{q}_i- \mathbf{q}_i(m)\|^2}{2\zeta(m+1)}\label{eq.priup1}\\
&\eta_i(m+1) = \argmin\limits_{\eta_i} \; \eta_i\frac{\partial f_i(m)}{\partial \eta_i}+\frac{\|\eta_i-\eta_i(m)\|^2}{2\zeta(m+1)}\label{eq.priup3}\\
&\mathbf{\bar{q}}(m+1) = \argmin\limits_{\mathbf{\bar{q}}} \sum_{i} \Big(-\bm\lambda_i^T(m)\mathbf{B}_i\mathbf{\bar{q}} 
+ \frac{\rho}{2}\|\mathbf{B}_i\mathbf{\bar{q}}-\mathbf{q}_i(m+1)\|^2\Big)\label{eq.priup4}
\end{align}
\end{subequations}
\begin{equation}\label{eq.dual}
\bm\lambda_i(m+1) = \bm\lambda_i(m) - \rho\big[\mathbf{B}_i\mathbf{\bar{q}}(m+1)-\mathbf{q}_i(m+1)\big]
\end{equation}

It is critical to note that the stochastic ADMM steps \eqref{eq.priup} and \eqref{eq.dual} can be implemented in a distributed fashion.
Specifically, each $\text{BS}_i$ maintains two vectors $\mathbf{q}_i$ and $\mathbf{\bar{q}}$. Given the knowledge of local CSI $\{\mathbf{R}_{ijk}\}_{j,k}$, $\text{BS}_i$ can solve the optimization problems \eqref{eq.priup1}-\eqref{eq.priup3}
independently, $\forall i$. After that, each BS broadcasts its latest $\mathbf{q}_i$ to other BSs. With the updated $\{\mathbf{q}_i\}$,
each BS can compute the public inter-BS interference vector $\mathbf{\bar{q}}$ according to \eqref{eq.priup4},
and then use it to update the dual variable $\bm{\lambda}_i$ by \eqref{eq.dual}.
The procedures of the proposed algorithm are summarized in Algorithm 2.

Recall that problem (\ref{tasknew}) as well as the reformulated (\ref{eq.task2}) are convex. As Algorithm 2 follows the stochastic ADMM steps, we readily have:
\begin{proposition}
The variables $\{\eta_i(m),P_i(m), \mathbf{q}_i(m)\}_{i=1}^{I}$, $\mathbf{\bar{q}}(m)$, and $\{\bm{\lambda}_i(m)\}_{i=1}^{I}$ in Algorithm 2
converge to the optimal primal and dual solutions of (\ref{eq.task2}) in expectation as $m \rightarrow \infty$.
When the algorithm converges, the beamforming matrices $\{\mathbf{W}_{i1}^*,\ldots,\mathbf{W}_{iK}^*\}_{i=1}^{I}$ obtained in Step 3
is a global optimal solution for the SDP problem \eqref{tasknew}.
In addition, we can always have: $\text{rank}(\mathbf{W}_{ik}^*)=1$, $\forall i,k$, when $\text{rank}(\mathbf{R}_{ijk})=1$, $\forall i, j, k$;
in this case, the optimal CMBF solution $\{\mathbf{w}_{ik}^*\}$ can be readily retrieved for the original problem (\ref{eq.task}).
\end{proposition}

\begin{proof}
Define $a_{\max}$ as the maximum possible electricity prices; i.e., we always have $b_i \leq a_i \leq a_{\max}$, $\forall i$.
It can be easily verified that $\mathbb{E}[\|f'(\mathbf{x},\boldsymbol{\xi})\|^2]\leq T$  [cf. \eqref{eq.fpi}-\eqref{eq.feta}],
$\forall x \in {\mathcal X}$, where $T=\frac{I(a_{\max}^2 + \theta^2)}{(1-\theta)^2}$ is a finite constant.
Given the boundedness of $\mathbb{E}[\|f'(\mathbf{x},\boldsymbol{\xi})\|^2]$, it then follows from \cite[Theorem 1]{hua13} that the proposed stochastic ADMM algorithm can converge to the optimal solution of the centralized problem \eqref{tasknew} in expectation,
as $m \rightarrow \infty$.
In other words, only ``stochastic'' convergence can be achieved; i.e., the  resultant primal and dual variables only hover within a small neighborhood around the optimal values, as will be verified by the simulation results in the sequel.

When $\mathbf{R}_{ijk}$ is rank-one, i.e., $\mathbf{R}_{ijk}=\mathbf{h}_{ijk}\mathbf{h}_{ijk}^H$, $\forall i,j,k$, we can follow the similar lines in the proof of Lemma 2 to show that problem (\ref{sadmm}) can be formulated into an SOCP; thus we can always have rank-one optimal $\{\mathbf{W}_{ik}^*\}_{i,k}$.
\end{proof}

\begin{remark}\textit{(Distributed inter-BS interference regularization)}
Algorithm 2 can be interpreted as an adaptive inter-BS interference regularization strategy where the coordinated BSs
gradually obtain their own beamforming solutions in an offline fashion.
The optimal $\{\mathbf{W}_{ik}^*\}$ are obtained until a consensus on the inter-BS interference powers among BSs is reached,
i.e., $\mathbf{B}_i\mathbf{\bar{q}}(m)=\mathbf{q}_i(m)$, $\forall i$.
For the time-invariant channel case where $\mathbf{R}_{ijk}=\mathbf{h}_{ijk}\mathbf{h}_{ijk}^H$, $\forall i,j,k$,
it is established that we always have the rank-one optimal matrices $\{\mathbf{W}_{ik}^*\}$.
For the general case where $\mathbf{R}_{ijk}$ can be of any rank, the tightness of the SDR could not be proved;
however, again our extensive numerical results indicate that this actually always holds even in this distributed CMBF scenario.
Note that the optimal CMBF solutions $\{\mathbf{w}_{ik}^*\}$ need only to be calculated
from $\{\mathbf{W}_{ik}^*\}$ upon convergence.
Proposition 1 then ensures that Algorithm 2 can yield the globally optimal CMBF scheme in a distributed manner.
In practice, the algorithm is re-run each time the statistical characteristics of the random variables
(energy harvesting amounts and electricity prices) change, or the propagation environment for the wireless communications
(channel covariance matrices) change.
The interval can range from tens of seconds to hours.

\end{remark}

\begin{remark}\textit{(Complexity and information exchange)}
Per iteration of Algorithm 2, each BS solves \eqref{eq.priup1} and \eqref{eq.priup3}. The sub-problem \eqref{eq.priup3} is a (convex) quadratic program which can be solved in closed-form. The sub-problem \eqref{eq.priup1} is essentially an SDP, which can be solved by general interior-point methods with a worst-case computational complexity $O\{[(I+N_t^2)K]^{3.5}\}$. On the other hand, each BS only needs to exchange with other BSs the vector $\mathbf{q}_i$, which contains its $IK$ local inter-BS interference levels. Clearly, the information exchange amount is limited and the computational complexity is affordable at each BS. In addition, it is established in \cite{hua13} that the algorithm can converge at a rate of ${O}(1/\sqrt{m})$ 
with regard to both objective values and feasibility violations; i.e., the proposed distributed scheme can quickly find the optimal CMBF solution.
\end{remark}

\section{Numerical Results}\label{sec:sim}
%

In this section, simulated tests are presented to evaluate the performance of the proposed algorithm.
We consider a coordinated multicell downlink with $I=4$ BSs, each having $N_t=8$ antennas, and serving $K=4$ single-antenna mobile users.
The covariance matrices $\{\mathbf{R}_{jik}\}, \forall j, i, k$ for wireless channels are chosen according to
the exponential correlation model in \cite{jlee10, zeng12}.
The $(m,n)$-th element of the matrix $\mathbf{R}_{jik}$ is given by $\alpha^{|m-n|}e^{j\beta(m-n)}$,
where $\alpha \in [0,1)$ (set as $0.9$ by default in simulations) is the correlation coefficient, $\beta \in [0,2\pi)$ is the phase difference between antennas
(varying across users), and with a little abuse of notation $j=\sqrt{-1}$.
Considering the large-scale fading, channel gain from the same cell (BS) is normalized to $1$, while those from other cells are multiplied by $0.25$,
as a fading coefficient.
The SINR threshold is set as $\gamma_{ik} = 8$ for all users.
The energy purchase price $a_i$ obeys a uniform distribution with a mean of $1 (\$/\text{KWh})$,
while the selling price is set as $b_i = 0.9a_i$, $i \in \cal I$.
Stochasticity of RES amount is mainly due to wind speed, thus we assume that $e_i$ follows the Weibull distribution,
which performs well to approximate the characteristics of wind speed variation \cite{yeh08, wen15}.
Constant penalty parameter $\rho = 1$ and stepsize $\zeta = 0.1$ are adopted in the proposed algorithm.
The CVaR confidence level is set as $\theta = 0.9$ unless otherwise stated.

Convergence of the proposed stochastic ADMM scheme is verified by Fig. \ref{obj}.
It is shown that the proposed distributed algorithm converges to the optimal solution obtained by the centralized algorithm within $200$ iterations.
Due to the convergence in expectation result established in Proposition 1, we can observe that upon convergence, the cost with the proposed stochastic scheme hovers within a small neighborhood around the optimal value provided by the centralized scheme.

\begin{figure}[t]
\centering
\vspace{-0.1cm}
\includegraphics[width=0.60\textwidth]{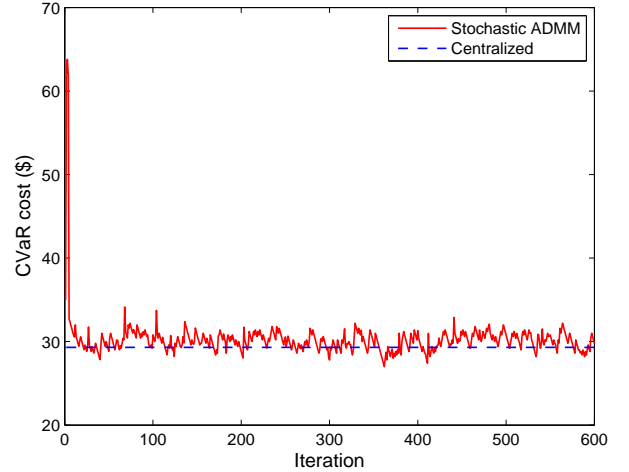}
\vspace{-0.65cm}
\caption{Convergence of the proposed stochastic ADMM scheme.}
\label{obj}
\vspace{-0.3cm}
\end{figure}

To gauge the performance of the proposed scheme, we use i) the optimal CMBF scheme for weighted sum-power minimization without RES (labeled as ``No RES'') \cite{shen12}, and ii) the optimal CMBF scheme for expected total-cost minimization when $\theta=0$ (labeled as ``Min-Cost''), as the baseline schemes. We include the performance of the proposed scheme (labeled as ``Min-CVaR''), when $\theta$ is set as $0.9$ or $0.5$,
and when the (random) RES amount $e_i$ obeys the Weibull or exponential distribution per $\text{BS}_i$.
Fig. \ref{avgcost} compares the average total energy transaction costs for different schemes.  It is clearly observed that the energy cost reduces for the wireless multicell downlink with RES integration, and the reduction becomes more significant as the available RES amount grows. When the average RES amount is $7.5$ KW per BS, the proposed scheme can save $\$27.5$ (or $57\%$) on average over the CMBF scheme without RES. We can also see that the distribution of RES does not affect significantly the optimal solution of the proposed scheme, which is very close to the one for expected total-cost minimization. This demonstrates the efficiency of the proposed scheme in reducing the energy transaction cost.

\begin{figure}[t]
\centering
\vspace{-0.1cm}
\includegraphics[width=0.60\textwidth]{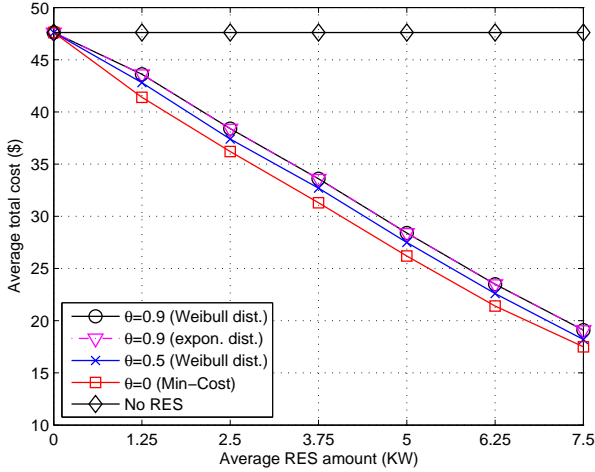}
\vspace{-0.65cm}
\caption{Average costs subject to different RES amounts.}
\label{avgcost}
\vspace{-0.3cm}
\end{figure}

Fig. \ref{cdf} depicts the empirical cumulative distribution functions (CDFs) of the energy transaction costs for the proposed scheme
under different values of the confidence level $\theta$, when average RES amount is $3.75$ KW per BS.
It is shown that use of CVaR, especially with a larger value of $\theta$, can effectively control the risk of very large energy transaction costs.
In particular, while the worst-case cost for the $\theta=0$ case (i.e., Min-Cost case) is $\$94.7$,
the worst-case costs for the $\theta=0.3$, $0.6$, and $0.9$ cases are
$\$73.2$, $\$68.7$, and $\$63.8$
(a $23\%$, $27\%$, and $33\%$ reduction), respectively.
Clearly, the proposed scheme is both robust and efficient, taking full advantage of stochastic RES.

\begin{figure}[t]
\centering
\vspace{-0.1cm}
\includegraphics[width=0.60\textwidth]{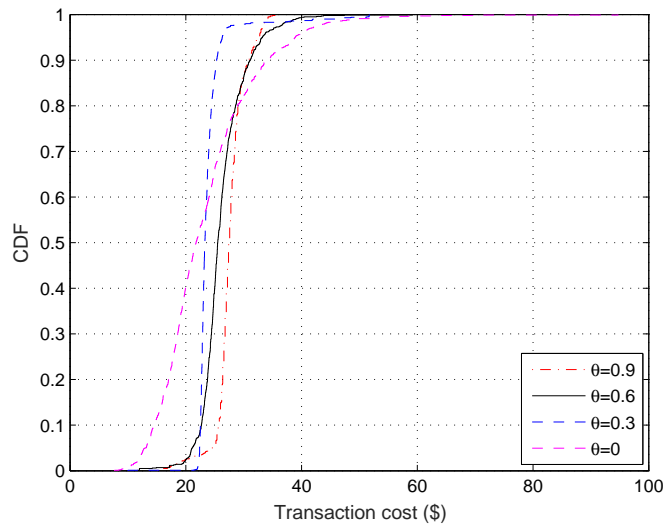}
\vspace{-0.65cm}
\caption{CDFs against different confidence levels.}
\label{cdf}
\vspace{-0.3cm}
\end{figure}

We further include the performance of the proposed ``Min-CVaR'' scheme against the ``Min-Cost'' approach under different numbers of transmit antennas in Fig. \ref{threent}. The average and worst-case costs of the two comparing schemes in different scenarios are listed in Table \ref{tabavg}.
It can be clearly observed that it costs less for the wireless system with an increasing number of transmit antennas.
When $N_t=8$, the proposed ``Min-CVaR'' scheme reduces the worst-case cost by $33\%$ with $7\%$ addition in average total cost
when compared with the ``Min-Cost'' scheme.
The ``Min-CVaR'' scheme is most efficient and can yield a much better performance than the ``Min-Cost'' one when $N_t=12$,
for it reduces the worst-case cost by $38\%$ with only $3\%$ addition in average total cost.
On the other hand, performances of the two schemes are close when $N_t=16$,
since there are sufficient transmit antennas serving a total of $16$ user equipments.
Merits of the proposed risk-constrained CMBF approach are clearly seen.

\begin{figure}[t]
\centering
\vspace{-0.1cm}
\includegraphics[width=0.60\textwidth]{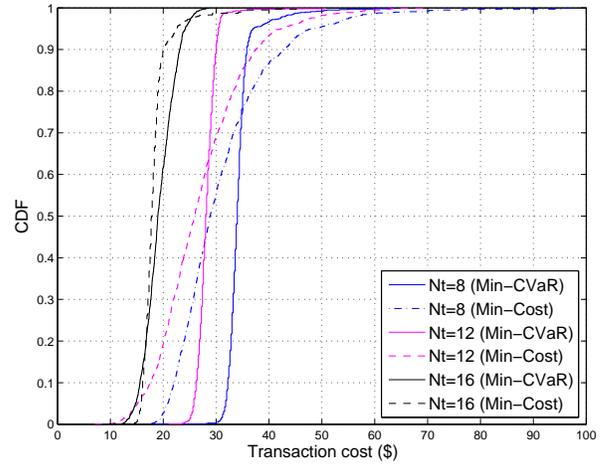}
\vspace{-0.65cm}
\caption{CDFs under different number of transmit antennas.}
\label{threent}
\vspace{-0.3cm}
\end{figure}

\begin{table}[t]
\renewcommand{\arraystretch}{1.2}
\centering
\caption{Average and worst-case costs $(\$)$ of different schemes under various values of $N_t$}\label{tabavg}
    \begin{tabular}{  | c | c | c | c | }
    \hline
    No. of transmit antennas ($N_t$)            & 8    &  12   &  16    \\ \hline

     Average cost (Min-CVaR)          & 33.5  & 23.5    &18.3     \\ \hline
     Average cost (Min-Cost)          & 31.2   & 22.9    &17.7    \\ \hline
     Worst-case cost (Min-CVaR)          & 63.8   & 43.6    &29.5    \\ \hline
     Worst-case cost (Min-Cost)          & 94.7   & 70.8    &44.2    \\ \hline
    \end{tabular}
\end{table}

Finally, we depict the influence of the SINR threshold on the average and worst-case transaction costs (when $N_t=8$) in Fig. \ref{sinr}.
It can be observed that the total cost increases along with the target SINR value.
The ``Min-CVaR'' scheme can control the worst-case cost efficiently when compared with the ``Min-Cost'' scheme,
especially in the high SINR scenario.
For instance, with the target SINR value $\gamma_{ik}=10$, the ``Min-CVaR'' scheme can reduce as much as $\$36$ of the worst-case cost,
with only $\$1.6$ addition in average cost.
Merits of the proposed approach are once again validated.

\begin{figure}[t]
\centering
\vspace{-0.1cm}
\includegraphics[width=0.60\textwidth]{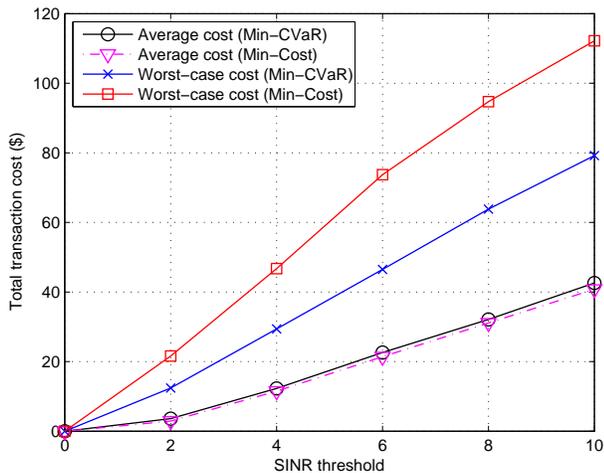}
\vspace{-0.65cm}
\caption{Average and worst-case costs under different SINR thresholds.}
\label{sinr}
\vspace{-0.3cm}
\end{figure}

Although the proof-of-concept simulation tests here are run for a small 4-BS system, the proposed CVaR-based CMBF design already displays great efficiency in controlling the risk of very high cost by saving as much as $\$30$ in the worst-case scenario.
Considering approximately a total of $6$ million BSs deployed across Chinese mainland \cite{web},
the proposed scheme could lead to a very large saving
($\$30/4 * 6 \text{ million} = \$45$ million) in electricity bill.
It can also help maintain a stable level of power flow from the electricity grid nation-wide,
and facilitates the development of ``green'' communication networks.


\section{Conclusions}\label{sec:con}

Distributed CMBF design was addressed for a RES-powered coordinated multicell downlink system.
The task was formulated into a convex program that minimizes the system-wide CVaR-based energy transaction cost with user QoS guarantees.
Leveraging the stochastic ADMM, optimal CMBF solution was obtained in a fast and distributed fashion.
Extensive tests corroborated the efficient, robust, risk-constrained and stability-maintaining merits of the proposed scheme.
The proposed framework paves a way to further advancing fundamental research on distributed resource allocation for next-generation integrated communication systems with RES.

\balance

\end{document}